\documentclass{ifacconf}%
\usepackage{graphicx}
\usepackage{natbib}
\usepackage{mathptmx}
\usepackage{amsmath}
\usepackage{amssymb}
\usepackage{graphicx}
\usepackage{amssymb}
\usepackage{epstopdf}
\usepackage{subfig}
\usepackage{multirow}%
\usepackage{amsfonts}

\providecommand{\U}[1]{\protect\rule{.1in}{.1in}}

\newtheorem{theorem}{Theorem}

\newtheorem{definition}{Definition}
\newtheorem{example}{Example}

\newtheorem{proof}{Proof}
\begin{document}
\begin{frontmatter}
\title{Stability of Fractional Order Switching Systems \thanksref{footnoteinfo}}
\thanks[footnoteinfo]{This work was supported by the Spanish Ministry of Science and Innovation under the project DPI2009-13438-C03.}
\author[First]{S. Hassan HosseinNia},
\author[First]{In\'es Tejado},
\author[First]{Blas M. Vinagre}
\address[First]{Department of Electrical, Electronics and Automation Engineering, School of Engineering, University of Extremadura, Badajoz, Spain\\(e-mail: \{hoseinia;itejbal;bvinagre\}@unex.es)}
\begin{abstract}                % Abstract of not more than 250 words.
This paper addresses the stabilization issue for fractional order switching systems. Common Lyapunov method is generalized for fractional order systems and frequency domain stability equivalent to this method is proposed to prove the quadratic stability. Some examples are given to show the applicability and effectiveness of the proposed theory.
\end{abstract}
\end{frontmatter}
%===============================================================================

\section{Introduction}

The past decade has witnessed an enormous interest in switched systems whose
behaviour can be described mathematically using a mixture of logic based
switching and difference/differential equations. By a switched system we mean
a hybrid dynamical system consisting of a family of continuous-time subsystems
and a rule that orchestrates the switching among them  (\cite{Liberzon03},
\cite{Daafouz02}). A primary motivation for studying such systems came partly
from the fact that switched systems and switched multi-controller systems have
numerous applications in control of mechanical systems, process control,
automotive industry, power systems, traffic control, and so on. In addition,
there exists a large class of nonlinear systems which can be stabilized by
switching control schemes, but cannot be stabilized by any continuous static
state feedback control law \cite{Lin09}.

Recent efforts in switched system research typically focus on the analysis of
dynamic behaviors, such as stability, controllability and observability, and
aim to design controllers with guaranteed stability and optimized performance
(refer to \cite{Lin09}, \cite{Shorten07} for a survey in recent results in the
field). To be more precise, the study of the stability issues of switched
systems gives rise to a number of interesting and challenging mathematical
problems, which have been of increasing interest in the recent decade.

Typically, the approach adopted to analyze these systems is to employ theories
that have been developed for differential equations. To this respect, most results are based on Lyapunov's stability theory which has played a dominant role in the analysis of dynamical systems for more than a century. Existence of quadratic Lyapunov functions for each of the constituent LTI systems is not sufficient for the stability of switched systems. However, it is well known that the switched system is stable if there exists some common Lyapunov function that satisfies the conditions of the Lyapunov theory simultaneously for all constituent subsystems (see e.g.~\cite{Liberzon03, Narendra_94, Mori_98, Shim_98}. Although \cite{Molchanov_89} established a number of converse theorems, showing that such common Lyapunov function always exists when the switched linear system is stable for arbitrary switching, general conditions for determining the existence of a common Lyapunov function for switched systems are unknown. Likewise, a frequency domain method equivalent to the common Lyapunov one may make the control and stability analysis easier. For example, \cite{Karimi_11} propose a frequency domain equivalent of common Lyapunov function based on strictly positive realness (SPR) of the system in order to analyze the quadratic stability of switching systems.

Given this context, the contribution of our work is to bring together theories
from several areas of control and to present stability issues in a unified
manner for fractional order switching systems. 

The remainder of this paper is organized as follows. Section \ref{sec_pre}
provides a collection of important issues concerning stability of switched
systems. The main contribution of this paper is presented in Section
\ref{sec_fos}, i.e., the stability theory developed for fractional order
switched systems. Section \ref{sec_examples} gives some examples to show the
applicability and goodness of the proposed stability issues. Section
\ref{sec_conclu} draws the concluding remarks.

\section{Preliminaries}\label{sec_pre}

When a system becomes unstable, the output of the system goes to infinity (or
negative infinity), which often poses a security problem in the immediate
vicinity. Also, systems which become unstable often incur a certain amount of
physical damage, which can become costly. For the sake of clarity, a
collection of important issues concerning stability of switched systems is
given in this section, mainly using Lyapunov theory.

\subsection{Stability theorems and basic definitions}

The idea behind Lyapunov's stability theory is as follows: assume there exists
a positive definite function with a unique minimum at the equilibrium. One can
think of such a function as a generalized description of the energy of the
system. If we perturb the state from its equilibrium, the energy will
initially rise. If the energy of the system constantly decreases along the
solution of the autonomous system, it will eventually bring the state back to
the equilibrium. Such functions are called Lyapunov functions. While Lyapunov
theorems generalize to nonlinear systems and locally stable equilibria we
shall only state them in the form applicable to our system class. Consider an
autonomous nonlinear dynamical system%
\begin{equation}
\dot{x}\left(  t\right)  =f(x\left(  t\right)  ),\ x\left(  0\right)
=x_{0},\label{Sys}%
\end{equation}
where $x(t)\in\mathcal{D}\subseteq\mathbb{R}^{n}$ denotes the system state
vector, $\mathcal{D}$ an open set containing the origin, and $f:\mathcal{D}%
\rightarrow\mathbb{R}^{n}$ continuous on $\mathcal{D}$. Suppose $f$ has an
equilibrium; without loss of generality, we may assume that it is at origin.
Then, Lyapunov stability for continuos systems can be summarized in the
following theorems.

\begin{theorem}
Let $x=0$ be an equilibrium point of (\ref{Sys}). Assume that there exists an
open set $\mathcal{D}$ with $0\in\mathcal{D}$ and a continuously
differentiable function $V:\mathcal{D}\rightarrow\mathbb{R}$ such that:

\begin{enumerate}
\item $V(0)=0,$

\item $V(x)>0$ for all $x\in\mathcal{D} \backslash\{0\}$, and

\item $\frac{\partial V}{\partial x}(x)f(x)\leq0$ for all $x\in\mathcal{D}.$
\end{enumerate}
then $x=0$ is a stable equilibrium point of (\ref{Sys}).
\end{theorem}

\begin{theorem}
If, in addition, $\frac{\partial V}{\partial x}(x)f(x)\leq0$ for all
$x\in\mathcal{D} \backslash\{0\}$, then $x=0$ is an asymptotically stable
equilibrium point.
\end{theorem}

\begin{definition}
[Quadratic Stability] A linear system
\begin{equation}
\dot{x}=Ax,\label{QSS}%
\end{equation}
is said to be quadratically stable in $\mathbb{R}$ if there exists a positive
definite matrix $P\in\mathbb{R}^{n\times n}$ such that,
\[
A^{T}P+PA<0.
\]
\label{LMIQS}
\end{definition}

\begin{definition}
[$t^{-a }$ Stability] The trajectory $x(t)=0$ of the system $\frac{d^{\alpha}%
x(t)}{dt^{\alpha}}=f(t,x(t)) $ is $t^{-a }$ asymptotically stable if the
uniform asymptotic stability condition is met and if there is a positive real
$a $ such that :\newline$\forall\left\Vert x\left(  t\right)  \right\Vert
,t\leq t_{0}\ \exists\ N\left(  x\left(  t\right)  ,t\leq t_{0}\right)
,\ t_{1}\left(  x\left(  t\right)  ,t\leq t_{0}\right)$ such that $\forall
t\leq t_{0},\ \left\Vert x\left(  t\right)  \right\Vert \leq\ N\left(
t-t_{1}\right)  ^{-a }.$
\end{definition}

$t^{-a }$ stability will thus be used to refer to the asymptotic stability of
fractional systems. The fact that the components of the state $x(t)$ 
decay slowly towards $0$ following $t^{-a }$ leads to fractional systems sometimes
being treated as long memory systems.

Let us consider a fractional order linear time invariant (FO-LTI) system as:
\begin{equation}
D^{\alpha}x=Ax, x\in\mathbb{R}^{n}\label{FOLTI}%
\end{equation}
where $\alpha$ is the fractional order.

\begin{theorem}
[\cite{Moze}]A fractional system given by (\ref{FOLTI}) with order $\alpha$, $1
\leq\alpha< 2$, is $t^{-a}$ asymptotically stable if and only if there exists a
matrix $P=P^{T} > 0$, $P \in\mathbb{R}^{n \times n}$, such that
\begin{equation}
\label{FSQ}%
\begin{bmatrix}
\left(  A^{T}P+PA \right) \sin\left(\phi \right)  & \left(
A^{T}P-PA \right) \cos\left( \phi \right) \\
\left(  -A^{T}P+PA \right) \cos\left(\phi \right)  & \left(
A^{T}P+PA \right) \sin\left(\phi \right)
\end{bmatrix}<0,
\end{equation}

where $\phi=\frac{\alpha\pi}{2}$.
\label{Moze}
\end{theorem}

\begin{theorem}
[\cite{Moze}] A fractional order system given by (\ref{FOLTI}) with order $\alpha$,
$0<\alpha\leq1$, is $t^{-a }$ asymptotically stable if and
only if there exists a positive definite matrix $P\in\mathbb{R}^{n}$ such that
\begin{equation}
\left(  -\left(  -A\right)  ^{\frac{1}{2-\alpha}}\right)  ^{T}P+P\left(
-\left( -A\right)  ^{\frac{1}{2-\alpha}}\right)  <0.
\end{equation}
\label{Moze0}
\end{theorem}

\subsection{Common quadratic Lyapunov functions}

Consider a switched system as follows:
\begin{equation}
\dot{x}=Ax, A \in co\left\{ A_{1}, ..., A_{L} \right\} ,
\label{Conv}
\end{equation}
where $"co"$ denotes the convex combination and $A_{i}, i=1,...,L$ is the switching subsystem. According to \cite{Pardalos_87},  (\ref{Conv}) can be alternatively written as:
\begin{equation}
\label{SWHM}\dot{x}=Ax, A=\sum_{i=1}^{L} \lambda_{i} A_{i}, \forall\lambda_{i}
\geq0, \sum_{i=1}^{L} \lambda_{i}=1.
\end{equation}

\begin{theorem}
[\cite{Boyd_94}]A system given by (\ref{SWHM}) is quadratically stable if
and only if there exists a matrix $P=P^{T} > 0$, $P \in\mathbb{R}^{n}\times
n$, such that
\[
\label{SWST}A_{i}^{T}P+PA_{i}<0, \forall i=1, ..., L.
\]
\end{theorem}

\subsection{Quadratic stability in frequency domain}

\cite{Karimi_11} propose an equivalent to common Lyapunov stability conditions in frequency domain. The relation between SPRness and the quadratic stability can be stated in the following theorem. For further information about the specification of state space system, refer to Section~\ref{sec_fos}.

\begin{theorem} [\cite{Karimi_11}]Consider $c_{1}(s)$ and $c_{2}(s)$, two stable polynomials of order $n$, corresponding to the systems $\dot{x}=A_{1}x$ and
$\dot{x}=A_{2}x$, respectively, then the following statements are equivalent:

\begin{enumerate}
\item {$\frac{c_{1}(s)}{c_{2}(s)}$ and $\frac{c_{2}(s)}{c_{1}(s)}$ are SPR.}

\item {$\left|  \arg(c_{1}(j\omega)) - \arg(c_{2}(j\omega))\right|  <
\frac{\pi}{2}$ $\forall$ $\omega$.}

\item { $A_{1}$ and $A_{2}$ are quadratically stable, which means that $\exists P
=P^{T} >0 \in\mathbb{R}^{n\times n}$ such that $A_{1}^{T}P+PA_{1} <0$ ,
$A_{2}^{T}P+PA_{2} <0$.}
\end{enumerate}

\label{Freq_stab}
\end{theorem}

\section{Quadratic stability of fractional order switching systems}\label{sec_fos}

This section will study two ways to obtain the quadratic stability of fractional order switching systems generalizing common Lyapunov functions for fractional order switching systems and obtaining an equivalent in frequency domain, respectively.

\subsection{Common quadratic Lyapunov functions of fractional order system}
Let us consider a fractional order switched system as: 
\begin{equation}
D^{\alpha}{x}=Ax, A \in co\left\{ A_{1}, ..., A_{L} \right\} ,\label{FSWHM}%
\end{equation}
where $\alpha$ is the fractional order.

\begin{theorem}
A fractional system described by (\ref{FSWHM}) with order $\alpha$, $1 \leq\alpha< 2$, is quadratically stable if and only if there exists a matrix $P=P^{T} > 0$, $P
\in\mathbb{R}^{n \times n}$, such that
\begin{align}%
\begin{bmatrix}
\left(  A_{i}^{T}P+PA_{i} \right) \sin\left( \phi \right)  &
\left(  A_{i}^{T}P-PA_{i} \right) \cos\left( \phi \right) \\
\left(  -A_{i}^{T}P+PA_{i} \right) \cos\left( \phi\right)  &
\left(  A_{i}^{T}P+PA_{i} \right) \sin\left( \phi \right)
\end{bmatrix}
<0,\nonumber\\
\forall i=1,..., L.
\end{align}
\label{FSQ}
\end{theorem}

\begin{proof}
System (\ref{FSWHM}) can be rewritten as:
\begin{equation}
D^{\alpha}{x}=Ax, A=\sum_{i=1}^{L} \lambda_{i} A_{i}, \forall\lambda_{i}
\geq0, \sum_{i=1}^{L} \lambda_{i}=1.
\end{equation}
Then, from Theorem {\ref{Moze}}, and (\ref{FSWHM}), we have
\begin{eqnarray*}
\begin{bmatrix}
\left (\mathcal{M}^TP+P\mathcal{M} \right )\sin(\phi) & \left ( \mathcal{M}^TP-P\mathcal{M} \right )\cos(\phi)\\ 
 \left (- \mathcal{M}^TP+P\mathcal{M} \right )\cos(\phi) & \left (\mathcal{M}^TP+P\mathcal{M} \right )\sin(\phi)
\end{bmatrix}, \\ \nonumber
\forall \lambda_i \geq 0, \sum_{i=1}^{L} \lambda_i=1\\ \nonumber
\Leftrightarrow \sum_{i=1}^{L} \lambda_i 
\left ( \begin{bmatrix}
\left ( A_i^TP+PA_i \right )\sin(\phi) & \left ( A_i^TP-PA_i \right )\cos(\phi)\\ 
\left ( -A_i^TP+PA_i \right )\cos(\phi) & \left ( A_i^TP+PA_i \right )\sin(\phi)
\end{bmatrix} \right ), \\ \forall \lambda_i \geq 0, \sum_{i=1}^{L} \lambda_i=1.
\end{eqnarray*}
where $\mathcal{M}=\sum_{i=1}^{L} \lambda_i A_i$ and $\phi=\frac{\alpha \pi}{2}$. Therefore, it is obvious that (\ref{FSWHM}) is quadratically stable if and only if
\begin{eqnarray}
\nonumber
\begin{bmatrix}
\left ( A_i^TP+PA_i \right )\sin(\phi) & \left ( A_i^TP-PA_i \right )\cos(\phi)\\ 
\left ( -A_i^TP+PA_i \right )\cos(\phi) & \left ( A_i^TP+PA_i \right )\sin(\phi)
\end{bmatrix}<0, \\ \forall i=1, ..., L.
\end{eqnarray}

\end{proof}

\begin{theorem}
A fractional system given by (\ref{FSWHM}) with order $\alpha$, $0<\alpha\leq1$, is
quadratically stable if and only if there exists a matrix $P=P^{T} > 0$, $P
\in\mathbb{R}^{n \times n}$, such that
\begin{equation}
\left(  -\left(  -A_{i}\right)  ^{\frac{1}{2-\alpha}}\right)  ^{T}P+P\left(
-\left( -A_{i}\right)  ^{\frac{1}{2-\alpha}}\right)  <0, \forall i=1, ..., L.
\end{equation}
\label{FSQ00}
\end{theorem}

\begin{proof}
Assume $\left [ I^{(1-\alpha)} x(t) \right ]_{t=0}=0$, the fractional order system (\ref{FSWHM}) with order $\alpha$, $0<\alpha\leq1$, can be replaced by the following integer order system
\cite{Moze}:
\begin{align}
\label{FSSL1}\dot{z}=A_{f}z , A_{f} \in Co\left\{ A_{f_{1}}, ..., A_{f_{L}}
\right\} \\
z=C_{f}x,
\end{align}
where $A_{f_{i}}=%
\begin{bmatrix}
0 & \cdots & 0 & A_{i}^{1/\alpha}\\
A_{i}^{1/\alpha} & \cdots & 0 & 0\\
& \ddots &  & \vdots\\
\cdots & 0 & A_{i}^{1/\alpha} & 0
\end{bmatrix}
$ and $C_{f}=%
\begin{bmatrix}
0 & \cdots & 0 & 1
\end{bmatrix}
$. Writing (\ref{FSSL1}) in an alternative way yields:
\begin{align}
\dot{z}=A_{f}z, A_{f}=\sum_{i=1}^{L} \lambda_{i} A_{f_{i}}, \forall\lambda_{i}
\geq0, \sum_{i=1}^{L} \lambda_{i}=1.
\end{align}
Therefore, assuming a positive definite matrix $\mathcal{P}>0$ with proper
size and, based on LMI method, the system (\ref{FSWHM}) with order $\alpha$, $0<\alpha\leq1$, is quadratically stable if:
\begin{align}
A^{T}_{f}\mathcal{P}+\mathcal{P}A_{f}<0 \Rightarrow\\
\sum_{i=1}^{L} \lambda_{i}(A^{T}_{f_{i}}\mathcal{P}+\mathcal{P}A_{f_{i}})<0
\Rightarrow\\
A^{T}_{f_{i}}\mathcal{P}+\mathcal{P}A_{f_{i}}<0, \forall i=1,...,L.\label{SCO}%
\end{align}
Then, it is obvious that expression (\ref{SCO}) is satisfied if and only if (\cite{Moze})
\begin{align}
(A^{1/\alpha}_{i})^{T}{P}+{P}A^{1/\alpha}_{i}<0, \forall
i=1,...,L,\label{SCOF}%
\end{align}
where $P$ is a positive definite matrix. In \cite{Moze} it is shown that condition
(\ref{SCOF}) is sufficient but not necessary to guarantee quadratic stability.
The necessary and sufficient condition for fractional order system is
given by Theorem \ref{Moze0}. Therefore, the necessary and sufficient
condition for fractional order system is
\begin{equation}
\left(  -\left(  -A_{i}\right)  ^{\frac{1}{2-\alpha}}\right)  ^{T}P+P\left(
-\left( -A_{i}\right)  ^{\frac{1}{2-\alpha}}\right)  <0, \forall i=1, ..., L.
\end{equation}
\end{proof}

\subsection{Frequency domain stability}
In this section, a link between quadratic stability using Lyapunov theory
and SPR properties will be provided, i.e., a connection between time domain and frequency domain conditions in order to obtain quadratic stability of fractional order switching systems. 

Consider a stable pseudo-polynomial of order $n\alpha$ as:
\begin{equation}
d(s)=s^{n\alpha}+d_{n-1}s^{(n-1)\alpha}+ \cdots+d_{1}s^{\alpha}+d_{0},
\end{equation}
which corresponds to the fractional order system $D^{\alpha}x=Ax$. Furthermore, consider a polynomial of order $2n$ as:
\begin{equation}
c(s)=s^{n}+c_{n-1}s^{(n-1)}+ \cdots+c_{1}s +c_{0},\label{TCP}%
\end{equation}
which corresponds to $\dot{\tilde{x}}=\tilde{A}\tilde{x}$. Assign
\begin{align}
C=%
\begin{bmatrix}
c_{2n-1}, & ..., & c_{1}, & c_{0}%
\end{bmatrix}
,\\
D=%
\begin{bmatrix}
d_{n-1}, & ..., & d_{1}, & d_{0}%
\end{bmatrix}
.
\end{align}
and
\begin{align}
\label{ab1}\tilde{A}=%
\begin{bmatrix}
-c_{n-1} & -c_{n-2} & \cdots & -c_{1} & -c_{0}\\
1 & 0 & \cdots & 0 & 0\\
0 & 1 & \cdots & 0 & 0\\
\vdots & \vdots & \ddots & \vdots & \vdots\\
0 & 0 & \cdots & 0 & 1\\
&  &  &  &
\end{bmatrix}
,\\
A=%
\begin{bmatrix}
-d_{n-1} & -d_{n-2} & \cdots & -d_{1} & -d_{0}\\
1 & 0 & \cdots & 0 & 0\\
0 & 1 & \cdots & 0 & 0\\
\vdots & \vdots & \ddots & \vdots & \vdots\\
0 & 0 & \cdots & 0 & 1\\
&  &  &  &
\end{bmatrix}
.
\end{align}

In the following, the necessary and sufficient condition for the quadratic
stability of fractional order switching systems will be given.

\begin{theorem}
Consider $d_{1}(s)$ and $d_{2}(s)$, two stable pseudo-polynomials of order
$n$, corresponding to the systems $D^{\alpha}x=A_{1}x$ and
$D^{\alpha}x=A_{2}x$ with order $\alpha$, $1\leq\alpha< 2$, respectively, then the following
statements are equivalent:

\begin{enumerate}
%\item{$\frac{d_1(s)}{d_2(s)}$ and $\frac{d_2(s)}{d_1(s)}$ are SPR.}

\item {%
\begin{align*}
\left|  \arg\left( \det((A_{1}^{2}-\omega^{2}I)-2(j\omega)A_{1}\sin
\frac{\alpha\pi}{2}) \right) - \right. \\
\left.  \arg\left(  \det((A_{2}^{2}-\omega^{2}I)-2(j\omega)A_{2}\sin
\frac{\alpha\pi}{2})\right) \right|  < \frac{\pi}{2}, \forall\omega.
\end{align*}
}

\item { $A_{1}$ and $A_{2}$ are quadratically stable meaning that: $\exists P
=P^{T} >0 \in\mathbb{R}^{n\times n}$ such that
\begin{align*}%
\begin{bmatrix}
\left(  A_{i}^{T}P+PA_{i} \right) \sin(\phi) & \left(
A_{i}^{T}P-PA_{i} \right) \cos(\phi)\\
\left(  -A_{i}^{T}P+PA_{i} \right) \cos(\phi) & \left(
A_{i}^{T}P+PA_{i} \right) \sin(\phi)
\end{bmatrix}
<0,\\
\forall i=1, 2.
\end{align*}
}
\end{enumerate}

\label{Freq_stab_frac}
\end{theorem}

\begin{proof}
Consider $c_{1}(s)$ and $c_{2}(s)$ are characteristic polynomials
corresponding to $\dot{\tilde{x}}=\tilde{A}_{1}\tilde{x}$ and $\dot{\tilde{x}%
}=\tilde{A}_{2}\tilde{x}$, respectively, where $\tilde{A}_{i}=%
\begin{bmatrix}
A_{i} \sin(\phi) & A_{i} \cos(\phi)\\
-A_{i} \cos(\phi) & A_{i} \sin(\phi)
\end{bmatrix}, i=1,2.$ According to Theorem \ref{Freq_stab}, the following statements are equivalents:

\begin{enumerate}
\item {$\frac{c_{1}(s)}{c_{2}(s)}$ and $\frac{c_{2}(s)}{c_{1}(s)}$ are SPR,
where $c_{i}(s)=\det(sI-\tilde{A}_{i}), i=1,2$.}

\item {$\left|  \arg(c_{1}(j\omega)) - \arg(c_{2}(j\omega))\right|  <
\frac{\pi}{2}$ $\forall$ $\omega$.}

\item { $\tilde{A}_{1}$ and $\tilde{A}_{2}$ are quadratically stable meaning
that: $\exists\mathcal{P} =\mathcal{P}^{T} >0 \in\mathbb{R}^{2n\times2n}$ such
that $\tilde{A}_{1}^{T} \mathcal{P}+ \mathcal{P}\tilde{A}_{1} <0$ , $\tilde
{A}_{2}^{T} \mathcal{P}+ \mathcal{P}\tilde{A}_{2} <0$.}
\end{enumerate}

Now, consider $d_{1}(s)$ and $d_{2}(s)$ are characteristic pseudo-polynomials
corresponding to the fractional order systems $D^{\alpha}x=A_{1}x$ and $D^{\alpha
}x=A_{2}x$ with order $\alpha$, $1 \leq\alpha< 2$, respectively. The relation between
$A_{i}$ and $\tilde{A}_{i}$ is given by (\ref{ab1}).
%Then, according to the Theorem \ref{Rei1} it is obvious that, $\frac{d_1(s)}{d_2(s)}$ and $\frac{d_2(s)}{d_1(s)}$ are SPR.
From (\ref{TCP}), we have
\begin{align*}
\left|  \arg(c_{1}(j\omega)) - \arg(c_{2}(j\omega))\right| =\\
\left|  \arg(j\omega I-\tilde{A_{1}}) - \arg(j\omega I-\tilde{A_{2}})\right|
< \frac{\pi}{2},\\
\Leftrightarrow\left|  \arg\left( \det((A_{1}^{2}-\omega^{2}I)-2(j\omega
)A_{1}\sin(\phi) \right) - \right. \\
\left.  \arg\left(  \det((A_{2}^{2}-\omega^{2}I)-2(j\omega)A_{2}\sin
(\phi)\right) \right|  < \frac{\pi}{2}, \forall\omega.
\end{align*}
where $I$ is the identity matrix with the proper size.

Define, $\mathcal{P} =%
\begin{bmatrix}
P & 0\\
0 & P
\end{bmatrix}
, P=P^{T}>0, P \in\mathbb{R}^{n\times n}$. Then,
\begin{align*}
\tilde{A}_{i}^{T} \mathcal{P}+ \mathcal{P}\tilde{A}_{i}=
\begin{bmatrix}
A_{i}^{T} \sin(\phi) & -A_{i}^{T} \cos(\phi)\\
A_{i}^{T} \cos(\phi) & A_{i}^{T} \sin(\phi)
\end{bmatrix}
\begin{bmatrix}
P & 0\\
0 & P
\end{bmatrix}
+\\%
\begin{bmatrix}
P & 0\\
0 & P
\end{bmatrix}
\begin{bmatrix}
A \sin(\phi) & A \cos(\phi)\\
-A \cos(\phi) & A \sin(\phi)
\end{bmatrix}
<0,\\
\Leftrightarrow%
\begin{bmatrix}
\left(  A_{i}^{T}P+PA_{i} \right) \sin(\phi) & \left(
A_{i}^{T}P-PA_{i} \right) \cos(\phi)\\
\left(  -A_{i}^{T}P+PA_{i} \right) \cos(\phi) & \left(
A_{i}^{T}P+PA_{i} \right) \sin(\phi)
\end{bmatrix}
<0,\\
\forall i=1, 2.
\end{align*}
\end{proof}

Therefore, the theorem is proved.

\begin{theorem}
Consider two stable fractional order systems $D^{\alpha}x=A_{1}x$ and
$D^{\alpha}x=A_{2}x$ with order $\alpha$, $0<\alpha\leq1$, then the following
statements are equivalent:

\begin{enumerate}
\item {$\left|  \arg(\det(\mathcal{A}_{1} -j\omega I) )-\arg(\det
(\mathcal{A}_{2} -j\omega I))\right|  < \frac{\pi}{2}$ $\forall$ $\omega$.}

\item {$A_{1}$ and $A_{2}$ are quadratically stable, which means that $\exists P
=P^{T} >0 \in\mathbb{R}^{n\times n}$ such that }%

\[
\left(  -\left(  -A_{i}\right)  ^{\frac{1}{2-\alpha}}\right)  ^{T}P+P\left(
-\left( -A_{i}\right)  ^{\frac{1}{2-\alpha}}\right)  <0, \forall i=1, 2,
\]
\end{enumerate}

where $\mathcal{A}_{i}=-\left( -A_{i}\right)  ^{\frac{1}{2-\alpha}}, \forall
i=1, 2$ and $I$ is the identity matrix. \label{Freq_stab_frac0}
\end{theorem}

\begin{proof}
Define $c_{i}(s)=\det(\mathcal{A}_{i}-sI)$, $i=1,2$. According to Theorem
\ref{Freq_stab} and common quadratic stability theorem for fractional order
system with order $\alpha$, $0<\alpha\leq1$, i.e., Theorem \ref{FSQ00}, proof is straightforward.
\end{proof}

Although the theory developed in the frequency domain no necessarily proves the SPRness, a relation was obtained as an equivalent issue of quadratic stability.
Concerning the ease of designing fractional order controllers in frequency
domain, the stability analysis in frequency domain will be really useful for
fractional order switching systems.

\section{Illustrative Examples}\label{sec_examples}

In this section, some examples are given in order to show the applicability of the  theories developed for fractional order switching systems.

\begin{example}
Let us consider the switching system (\ref{FSWHM}) with order $\alpha=0.6$,
where $A_{1}=%
\begin{bmatrix}
0.3529 & 1.6044\\
-1.6044 & -4.4602
\end{bmatrix}
$ and $A_{2}=%
\begin{bmatrix}
0.3661 & 0.9237\\
-0.4618 & -0.1558
\end{bmatrix}
$. Applying Theorem \ref{FSQ00} yields:
\begin{align*}
\mathcal{A}_{1}=%
\begin{bmatrix}
0 & 1\\
-1 & -3
\end{bmatrix}
, \mathcal{A}_{2}=%
\begin{bmatrix}
0 & 1\\
-0.5 & -0.5
\end{bmatrix}
\end{align*}
Then, choosing a common matrix $P=%
\begin{bmatrix}
3 & 1\\
1 & 4
\end{bmatrix}
$, the stability conditions
\begin{align*}
\left(  -\left(  -A_{1}\right)  ^{\frac{1}{2-\alpha}}\right)  ^{T}P+P\left(
-\left( -A_{1}\right)  ^{\frac{1}{2-\alpha}}\right)  =%
\begin{bmatrix}
-2 & -4\\
-4 & -22
\end{bmatrix}
<0\\
\left(  -\left(  -A_{2}\right)  ^{\frac{1}{2-\alpha}}\right)  ^{T}P+P\left(
-\left( -A_{2}\right)  ^{\frac{1}{2-\alpha}}\right)  =%
\begin{bmatrix}
-1 & 0.5\\
0.5 & -2
\end{bmatrix}
<0
\end{align*}
are satisfied and the switching system is quadratically stable. Now let us
compare the results with the frequency domain analysis. Applying Theorem
\ref{Freq_stab_frac0}, the following condition
\begin{align}
\left|  \tan^{-1}\left( \frac{3\omega}{1+\omega^{2}}\right) - \tan^{-1}\left(
\frac{0.5\omega}{0.5+\omega^{2}}\right)  \right|  < \frac{\pi}{2},
\forall\omega\label{Exp1_Freq}%
\end{align}
should be satisfied. The phase difference of (\ref{Exp1_Freq}) is depicted in Fig.
\ref{Exp1_fig}. As can be observed, the maximum phase difference is $51.51^{o}$, which is less than $90^{o}$, that implies the switching stability condition is satisfied and the system is quadratically stable. 
\begin{figure}[ptbh]
\begin{center}
\includegraphics[width=0.5\textwidth]{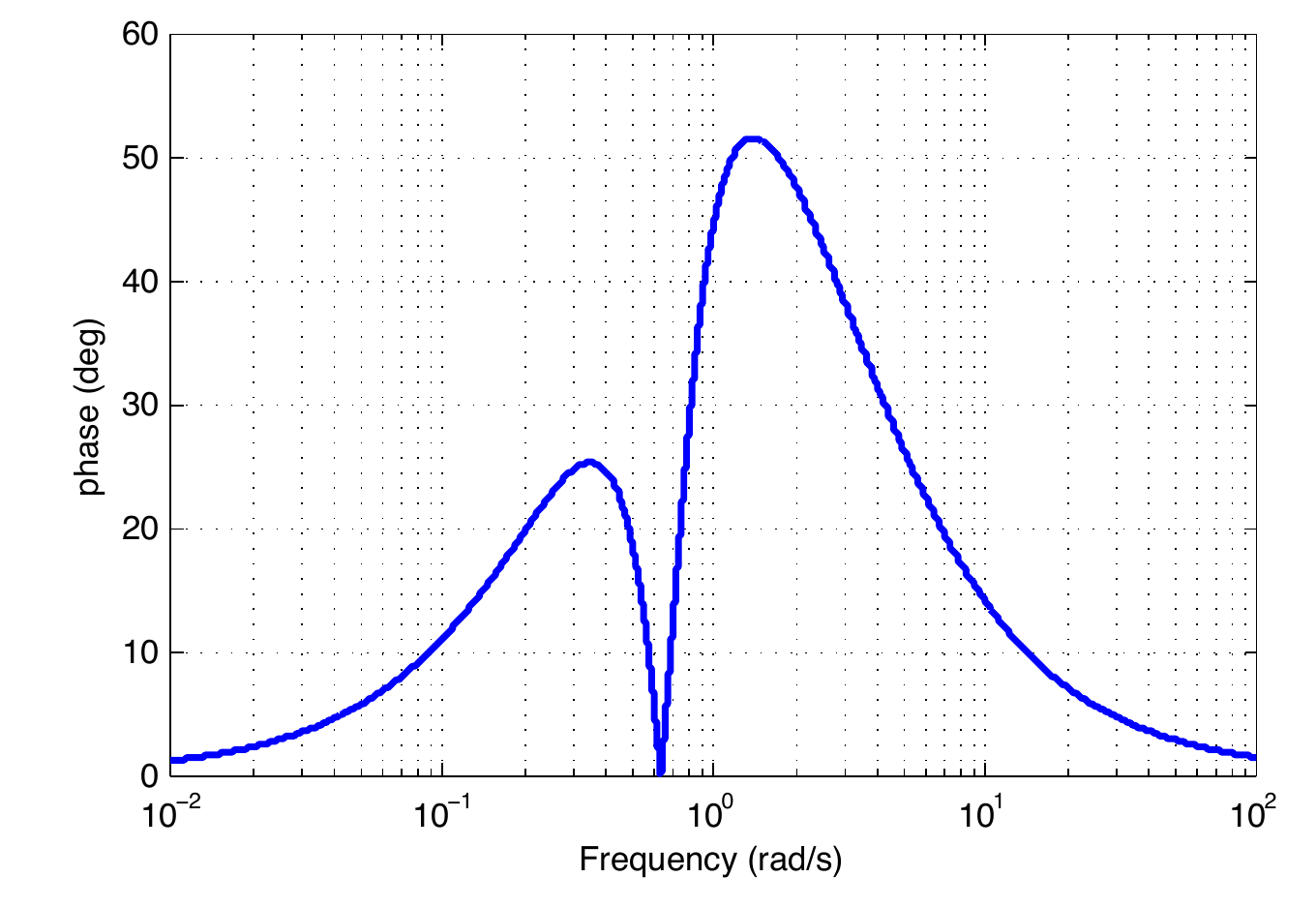}
\end{center}
\caption{Phase difference of condition (\ref{Exp1_Freq}) for the system in Example 1}%
\label{Exp1_fig}%
\end{figure}
\end{example}

\begin{example}
Now, let us consider the switching system given by (\ref{FSWHM}) with order $\alpha
=0.75$, where $A_{1}=%
\begin{bmatrix}
0 & 1\\
-1 & 0.5
\end{bmatrix}
$ and $A_{2}=%
\begin{bmatrix}
0 & 1\\
-0.5 & 0.1
\end{bmatrix}
$. Applying Theorem \ref{Freq_stab_frac0}, we have:
\begin{align*}
\mathcal{A}_{1}=%
\begin{bmatrix}
-0.3684 & 1.0263\\
-1.0263 & 0.1448
\end{bmatrix}
, \mathcal{A}_{2}=%
\begin{bmatrix}
-0.2450 & 1.0390\\
-0.5195 & -0.1411
\end{bmatrix}
\end{align*}
and the following frequency domain condition
\begin{align}
\left|  \tan^{-1}\left( \frac{0.2236\omega}{1+\omega^{2}}\right) - \tan
^{-1}\left( \frac{0.3861\omega}{0.4977+\omega^{2}}\right)  \right|  <
\frac{\pi}{2}, \forall\omega\label{Exp2_Freq}%
\end{align}
should be satisfied. The same as previous example, condition (\ref{Exp2_Freq}) is
depicted in Fig.~\ref{Exp2_fig}. It is shown that the maximum
phase difference of condition (\ref{Exp2_Freq}) is $80.02^{o}$, so the switching system is quadratically stable.

\begin{figure}[ptbh]
\begin{center}
\includegraphics[width=0.5\textwidth]{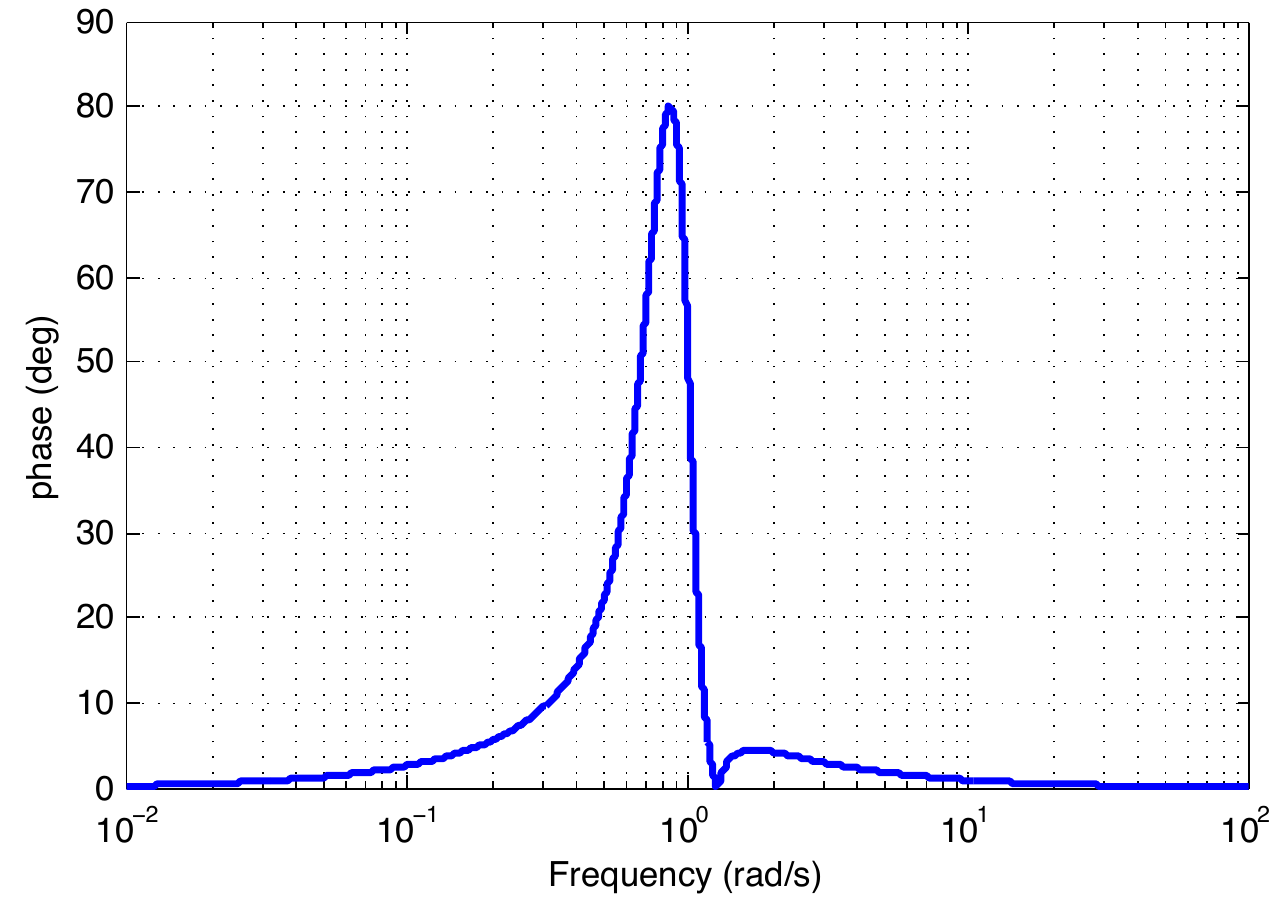}
\end{center}
\caption{Phase difference of condition (\ref{Exp2_Freq}) for the system in Example 2}%
\label{Exp2_fig}%
\end{figure}
\end{example}

\begin{example}
Let us now consider the same switching system as in Example 2, but with an order
bigger than $1$, $1<\alpha<2$. Applying Theorem \ref{Freq_stab_frac}, the following condition:
\begin{align}
\left|  \arg\left( \det\left(
\begin{bmatrix}
-1.75-\omega^{2} & -2.1-2j\sin(\phi)\\
3.675+3.5j\sin(\phi) & 2.66-\omega^{2}+4.2j\sin(\phi)%
\end{bmatrix}
\right)  \right) - \right. \nonumber\\
\left.  \arg\left( \det\left(
\begin{bmatrix}
-3-\omega^{2} & -3-2j\sin(\phi)\\
9+6j\sin(\phi) & 6-\omega^{2}+6j\sin(\phi)%
\end{bmatrix}
\right)  \right) \right|  < \frac{\pi}{2}, \forall\omega\label{exp3eq}%
\end{align}
where $\phi=\frac{\alpha\pi}{2}$ should be satisfied for all $\alpha$, $1<\alpha<2$. Figure~\ref{Exp3_Var} represents the condition (\ref{exp3eq}) when the fractional
order $\alpha$ is changing in the interval $(1,2)$. In order to make this example clearer, the interval of variation of $\alpha$ is divided into three subintervals. As a matter of fact, Fig.~\ref{Exp3_Var} (a) shows the phase difference (\ref{exp3eq}) for systems with the order $\alpha\in(1,1.5]\cup\lbrack1.7,2)$, whereas Fig.~\ref{Exp3_Var} (b) corresponds to systems with order $\alpha\in(1.5,1.7)$. As can be seen, the system is quadratically stable if its order $\alpha\in(1,1.5]\cup\lbrack1.7,2)$. The stability region of the considered system  is shown in Fig.~\ref{Exp3_max}, in which the maximum values of (\ref{exp3eq}) are plotted versus its order $\alpha$. 

\begin{figure}[ptbh]
\begin{center}
\includegraphics[width=0.5\textwidth]{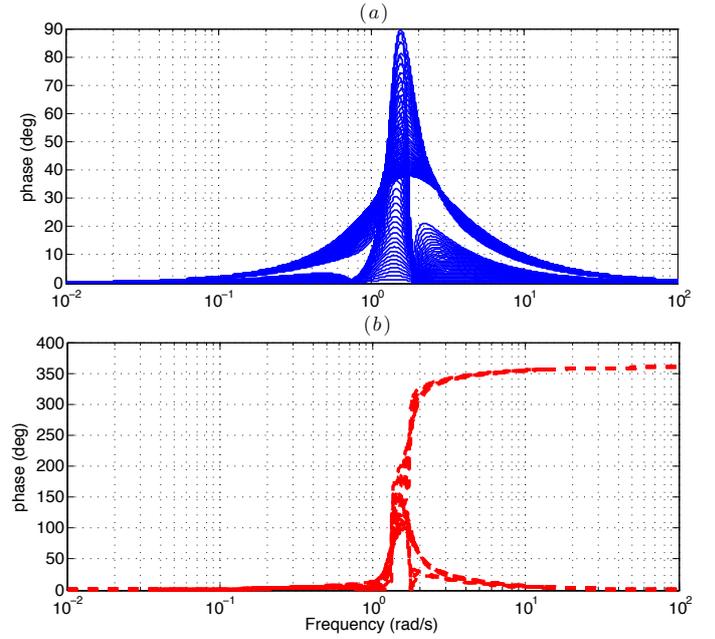}
\end{center}
\caption{Phase difference of condition (\ref{exp3eq}) for the system in Example 3 for different values of the order $\alpha$: (a) $\alpha\in(1,1.5]\cup\lbrack1.7,2)$ (b) $\alpha\in(1.5,1.7)$}%
\label{Exp3_Var}%
\end{figure}

\begin{figure}[ptbh]
\begin{center}
\includegraphics[width=0.5\textwidth]{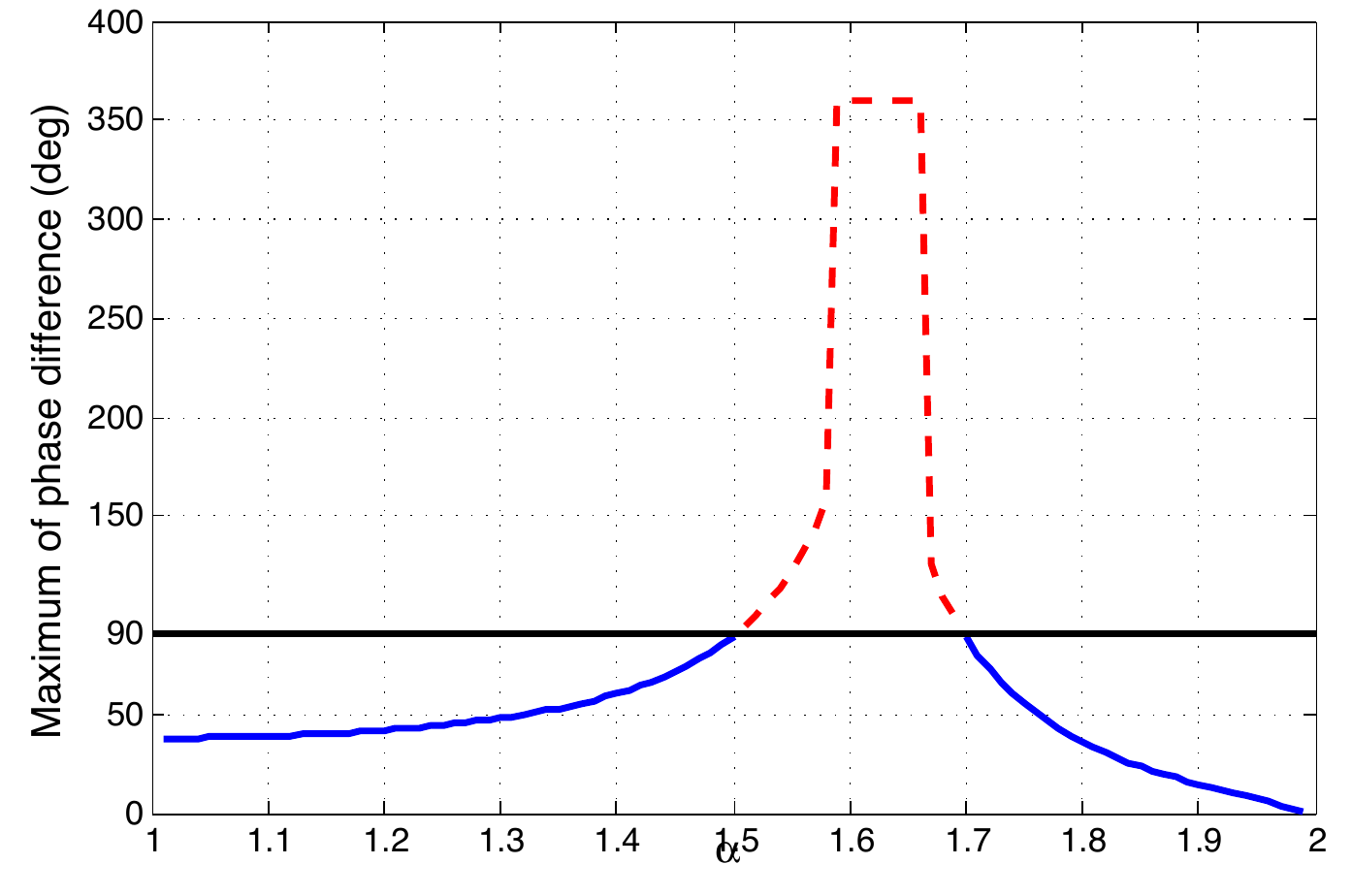}
\end{center}
\caption{Maximum value of (\ref{exp3eq}) versus $\alpha$}%
\label{Exp3_max}%
\end{figure}
\end{example}

\section{Conclusion}\label{sec_conclu}

This paper studies the quadratic stability for fractional order switching systems. In particular, equivalent Lyapunov conditions in frequency domain are developed for this kind of systems to prove their quadratic stability. Some illustrative examples are given to show the applicability and validation of the proposed theory.

Our future efforts will focus on finding a relation between the frequency domain method proposed in this paper and SPRness.

\bibliographystyle{plain}
\bibliography{ifacconf}

\begin{thebibliography}{13}
\providecommand{\natexlab}[1]{#1}
\providecommand{\url}[1]{\texttt{#1}}
\providecommand{\urlprefix}{URL }
\expandafter\ifx\csname urlstyle\endcsname\relax
  \providecommand{\doi}[1]{doi:\discretionary{}{}{}#1}\else
  \providecommand{\doi}{doi:\discretionary{}{}{}\begingroup
  \urlstyle{rm}\Url}\fi

\bibitem[{Boyd et~al.(1994)Boyd, Ghaoui, Feron, and Balakrishnan}]{Boyd_94}
Boyd, S., Ghaoui, L.E., Feron, E., and Balakrishnan, V. (1994).
\newblock \emph{Linear Matrix Inequalities in System and Control Theory}.
\newblock SIAM.

\bibitem[{Daafouz et~al.(2002)Daafouz, Riedinger, and Iung}]{Daafouz02}
Daafouz, J., Riedinger, P., and Iung, C. (2002).
\newblock Stability analysis and control synthesis for switched systems: {A}
  switched {L}yapunov function approach.
\newblock \emph{IEEE Transactions on Automatic Control}, 47(11), 1883 -- 1887.

\bibitem[{Kunze and Karimi(2011)}]{Karimi_11}
Kunze, M. and Karimi, A. (2011).
\newblock Frequency-domain controller design for switched systems.
\newblock \emph{Submitted to Automatica}.

\bibitem[{Liberzon(2003)}]{Liberzon03}
Liberzon, D. (2003).
\newblock \emph{Switching in Systems and Control}.
\newblock Birk\"{a}user.

\bibitem[{Lin and Antsaklis(2009)}]{Lin09}
Lin, H. and Antsaklis, P.J. (2009).
\newblock Stability and stabilizability of switched linear systems: {A} short
  survey of recent results.
\newblock \emph{IEEE Transactions on Automatic Control}, 54(2), 24 --29.

\bibitem[{Molchanov and Pyatnitskii(1989)}]{Molchanov_89}
Molchanov, A.P. and Pyatnitskii, E.S. (1989).
\newblock {C}riteria of asymptotic stability of differential and difference
  inclusions encountered in control theory.
\newblock \emph{Systems and Control Letters}.

\bibitem[{Mori et~al.(1998)Mori, Mori, and Kuroe}]{Mori_98}
Mori, Y., Mori, T., and Kuroe, Y. (1998).
\newblock {O}n a class of linear constant systems which have a common quadratic
  lyapunov function.
\newblock In \emph{Proceedings of the 37th IEEE Conference on Decision and Control}.

\bibitem[{Moze et~al.(2007)Moze, Sabatier, and Oustaloup}]{Moze}
Moze, M., Sabatier, J., and Oustaloup, A. (2007).
\newblock {LMI} characterization of fractional systems stability.
\newblock \emph{Advances in Fractional Calculus}, 419--434.

\bibitem[{Narendra and Balakrishnan(1994)}]{Narendra_94}
Narendra, K.S. and Balakrishnan, J. (1994).
\newblock A common Lyapunov function for stable LTI system with commuting
  a-matrices.
\newblock \emph{IEEE Transactions on Automatic Control}, 39(12), 2469--2471.

\bibitem[{Pardalos and Rosen(1987)}]{Pardalos_87}
Pardalos, P. and Rosen, J. (1987).
\newblock Constrained global optimization: Algorithms and applications.
\newblock \emph{Lecture Notes in Computer Science}, 268.

\bibitem[{Shim et~al.(1998)Shim, Noh, , and Seo}]{Shim_98}
Shim, H., Noh, D., , and Seo, J. (1998).
\newblock Common Lyapunov function for exponentially stable nonlinear systems.
\newblock In \emph{Proceedings of the 4th SIAM Conference on Control and its Applications}.

\bibitem[{Shorten et~al.(2007)Shorten, Wirth, Mason, Wulff, and
  King}]{Shorten07}
Shorten, R., Wirth, F., Mason, O., Wulff, K., and King, C. (2007).
\newblock Stability criteria for switched and hybrid systems.
\newblock \emph{SIAM Review}, 49(4), 545--592.

\end{thebibliography}
%bib file to produce the bibliography
%with bibtex (preferred)

%\begin{thebibliography}{xx}  % you can also add the bibliography by hand

%\bibitem[Able(1956)]{Abl:56}
%B.C. Able.
%\newblock Nucleic acid content of microscope.
%\newblock \emph{Nature}, 135:\penalty0 7--9, 1956.

%\bibitem[Able et~al.(1954)Able, Tagg, and Rush]{AbTaRu:54}
%B.C. Able, R.A. Tagg, and M.~Rush.
%\newblock Enzyme-catalyzed cellular transanimations.
%\newblock In A.F. Round, editor, \emph{Advances in Enzymology}, volume~2, pages
%125--247. Academic Press, New York, 3rd edition, 1954.

%\bibitem[Keohane(1958)]{Keo:58}
%\newblock \emph{Power and Interdependence: World Politics in Transitions}.
%\newblock Little, Brown \& Co., Boston, 1958.

%\bibitem[Powers(1985)]{Pow:85}
%\newblock Is there a way out?
%\newblock \emph{Harpers}, pages 35--47, June 1985.

%\bibitem[Soukhanov(1992)]{Heritage:92}
%\newblock \emph{{The American Heritage. Dictionary of the American Language}}.
%\newblock Houghton Mifflin Company, 1992.

%\end{thebibliography}

\end{document}